\newif\ifelsevierversion
\newif\ifarxivversion
\elsevierversionfalse 

\ifelsevierversion
  \arxivversionfalse
\else
  \arxivversiontrue
\fi

\ifelsevierversion
    \documentclass[preprint]{elsarticle}
\else
    \documentclass{article}
\fi

\usepackage{comment}

\usepackage{hyperref} 
\hypersetup{
    colorlinks,%
    citecolor=black,%
    filecolor=black,%
    linkcolor=black,%
    urlcolor=black
}

\usepackage{bm}
\usepackage{amssymb}
\usepackage{amsmath}
\usepackage{amsthm}
\newtheorem{theorem}{Theorem}
\newtheorem{definition}[theorem]{Definition}  

\newtheorem{lemma}[theorem]{Lemma}

\usepackage[textwidth=4.9in]{geometry}

\usepackage{paralist}
\usepackage{array}

\usepackage{tikz}

\newcommand{\tne}[1]{\ensuremath{\langle #1\rangle}}

\newcommand{\nopTape}{\ensuremath{\_}}
\newcommand{\mrkTape}{\ensuremath{M}}
\newcommand{\rgtTape}{\ensuremath{R}}
\newcommand{\lftTape}{\ensuremath{L}}

\newcommand{\TM}{\ensuremath{\mathcal{M}}}
\newcommand{\NETM}{\ensuremath{\mathcal{N}}}
\newcommand{\WBM}{\ensuremath{\mathcal{W}}}
\newcommand{\HM}{\ensuremath{\mathcal{H}}}

\ifelsevierversion
  \journal{Journal of Complexity}
\else
  \usepackage{authblk}
  \usepackage[active]{srcltx}
  \newcommand{\email}[1]{\url{#1}}
\fi

\begin{document}

\newcommand{\nmtitlenotetext}{This work is the outcome of a visit by the authors to the Isaac Newton Institute for Mathematical Sciences in Cambridge (UK) for the centenary of Alan Turing's birth, hosted by the 2012 programme ``Semantics and Syntax: A Legacy of Alan Turing'.}
\ifelsevierversion
  \newcommand{\nmtitlenote}{\tnoteref{ack}}
  \newcommand{\nmtitlenoter}{\tnotetext[ack]{\nmtitlenotetext}}

  \begin{frontmatter}
\fi 
\ifarxivversion
  \newcommand{\nmtitlenote}{\footnote{\nmtitlenotetext}}
  \newcommand{\nmtitlenoter}{}
  \renewcommand*{\thefootnote}{\fnsymbol{footnote}}
  \renewcommand\Authands{,~and~} 
\fi

\title{Wang's B machines are efficiently universal, as is Hasenjaeger's small universal electromechanical toy\nmtitlenote}
\nmtitlenoter

\ifelsevierversion
\else
  \author[1]{Turlough Neary\thanks{Turlough Neary was supported by Swiss National Science Foundation grant 200021-141029.}}
  \author[2]{Damien Woods\thanks{Damien Woods was supported by the USA National Science Foundation under grants 0832824 (The Molecular Programming Project), CCF-1219274, and CCF-1162589.}}
  \author[3]{Niall~Murphy\thanks{Niall Murphy was supported by a PICATA Young Doctors fellowship from CEI Campus Moncloa, UCM-UPM, Madrid, Spain.}}
  \author[5]{Rainer~Glaschick\thanks{Rainer Glaschick was supported by the Heinz Nixdorf MuseumsForum, Germany.}}  

  \affil[1]{Institute for Neuroinformatics, University of Z\"urich and ETH Z\"urich, Switzerland. \email{tneary@ini.phys.ethz.ch}}
  \affil[2]{California Institute of Technology, Pasadena, CA 91125, USA. \email{woods@caltech.edu}}
  \affil[3]{Facultad de Inform\'atica, Universidad Polit\'ecnica de Madrid and
    CEI-Moncloa UPM-UCM, Spain.}
  \affil[4]{(Present affiliation) Microsoft Research, Cambridge, CB1 2FB, UK. \email{a-nimurp@microsoft.com}}
  \affil[5]{Paderborn, Germany. \email{rainer@glaschick-pb.de}}
\fi
    
\ifarxivversion
  \date{}
  \maketitle
  \vspace{-6ex} 
\fi
  \begin{abstract}
    In the 1960's Gisbert Hasenjaeger built Turing Machines from electromechanical relays and uniselectors.
    Recently,
      Glaschick reverse engineered the program of one of these machines
        and found that it is a universal Turing machine.
    In fact, its  program  uses only four states and two symbols,
      making it a very small universal Turing machine.      
    (The machine has three tapes and a number of other features 
      that are important to keep in mind when comparing it to
        other small universal machines.)
    Hasenjaeger's machine simulates Hao Wang's~B machines, 
      which were proved universal by Wang.         
    Unfortunately, Wang's original simulation algorithm
      suffers from an exponential slowdown when simulating Turing machines. 
    Hence, via this simulation, Hasenjaeger's machine also has an exponential slowdown when simulating Turing machines.
    In this work, we give a new efficient simulation algorithm for Wang's B machines by showing that they simulate Turing machines with only a polynomial slowdown. 
    As a second result, we find that Hasenjaeger's machine also efficiently simulates Turing machines in polynomial time.
    Thus, Hasenjaeger's machine is both small and fast.
      In another application of our result, we show that Hooper's small
      universal Turing machine simulates Turing machines in polynomial time,
      an exponential improvement.
  \end{abstract}
\ifelsevierversion
  \end{frontmatter}
\fi

  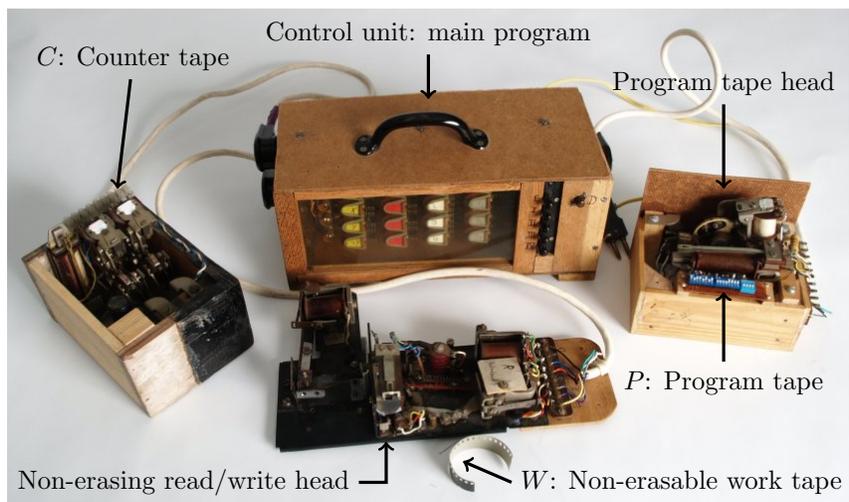
\begin{figure}[t]
    \centering
    \begin{tikzpicture}
      \node[anchor=south west,inner sep=0]
        (image) at (0,0) 
          {\includegraphics[width=0.9\textwidth]{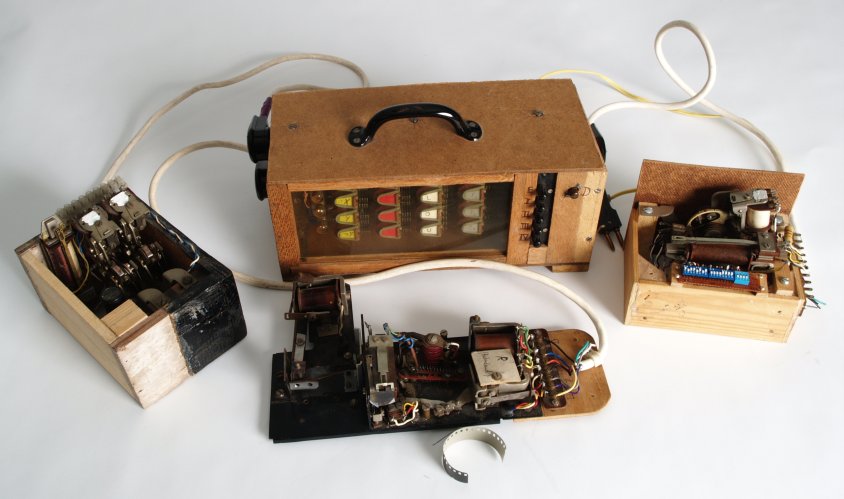}};
      \begin{scope}[x={(image.south east)},y={(image.north west)}]

        \node (l4stateproc) at (0.5, 0.95) {Control unit: main program};
        \node (p4stateproc) at (0.5, 0.8) {};
        \draw [->,very thick] (l4stateproc) -- (p4stateproc);

        \node (lprogtape) at (0.85, 0.85) {Program tape head};
        \node (pprogtape) at (0.85, 0.62) {};
        \draw [->,very thick] (lprogtape) -- (pprogtape);

        \node (l18bitprog) at (0.85, 0.25){$P$: Program tape};
        \node (p18bitprog) at (0.85, 0.45) {};
        \draw [->,very thick] (l18bitprog) -- (p18bitprog);

        \node (lrestape) at (0.8, 0.05){$W$: Non-erasable work tape};
        \node (prestape) at (0.53, 0.1) {};
        \draw [->,very thick] (lrestape.west) -- (prestape);
        
        \node (lrestaperw) at (0.21, 0.05){Non-erasing read/write head};
        \node (prestaperw) at (0.45, 0.15) {};
        \draw [->,very thick] (lrestaperw) -| (prestaperw);

        \node (lcounttape) at (0.15, 0.9){$C$: Counter tape};
        \node (pcounttape) at (0.13, 0.62) {};
        \draw [->,very thick] (lcounttape) -- (pcounttape);
      \end{scope}
    \end{tikzpicture}
    \caption{Hasenjaeger's universal electromechanical Turing machine. The wiring in the control unit encodes a universal program, that uses only four states and two symbols, for simulating Wang~B machines.
    The program of a Wang~B machine may be stored on the program tape. There are two additional tapes which are used for the simulation, a counter tape and a work tape.   }\label{fig:Hasenjaeger's UTM}
  \end{figure}

  \section{Introduction}\label{sec:intro}
  In the 1960's Gisbert Hasenjaeger built Turing Machines  from
  electromechanical relays and uniselectors, but never published details of 
  these machines.
  Recently, Hasenjaeger's family donated
    the machine shown in Figure~\ref{fig:Hasenjaeger's UTM} 
    to the Heinz Nixdorf MuseumsForum.\footnote{Heinz Nixdorf MuseumsForum, Paderborn, Germany. \url{http://www.hnf.de/}}
  At the request of the MuseumsForum,
  Glaschick reverse engineered the table
  of behaviour for this machine~\cite{Glaschick2012a,Glaschick2012b}, and, using
  Hasenjaeger's notes~\cite{Hasenjaeger1987}, determined the machine's
  encoding and operation. It was found that Hasenjaeger's machine
  simulates Wang's B machines~\cite{Wang1957}. 
 
  Wang used a unary encoding when proving his B machines universal and hence 
  they suffer from an exponential slowdown when simulating Turing machines.
  As a result, Hasenjaeger's machine also suffers from an exponential slowdown.
  In this work, we show that Wang~B machines and Hasenjaeger's machine
  simulate Turing machines with polynomial slowdown via the following chain of
  simulations:
  \begin{equation*}
  \begin{array}{c}
     \textrm{Turing Machine} \mapsto  \textrm{non-erasing Turing Machine} \mapsto \\ 
      \textrm{Wang B machine} \mapsto  \textrm{Hasenjaeger's universal Turing Machine} 
  \end{array}
  \end{equation*}
  where $A\mapsto B$ denotes that $A$ is simulated by $B$.
  With the exception of the Wang~B machine simulation of non-erasing machines,
  all of the simulations in the above chain are known to be efficient:
  non-erasing Turing machines simulate Turing machines with a polynomial
  slowdown in time~\cite{Zykin1963}, and Hasenjaeger's machine simulates Wang
  B machines in linear time.  We complete the chain of efficient simulations
  by giving a new simulation algorithm that shows that Wang's B machines
  simulate Turing machines with only a polynomial slowdown in the simulated
  Turing machine's time. An immediate consequence of our new algorithm is that
  Hasenjaeger's machine also simulates Turing machines in polynomial time.
  This adds to the growing body of work~\cite{WN2006c,NW2006,neary2012complexity}
  showing that the simplest known universal models of computation need not suffer from a
  exponential slowdown.

    As mentioned above, the simulation of Turing machines by non-erasing
    Turing machines is already known to run with a polynomial slowdown~\cite{Zykin1963}.  However, to keep our paper self-contained, we give
    our own polynomial (cubic) time  simulation in Section~\ref{sec:nonerasingTM}.
    This is followed by our main result in Section~\ref{sec:Wang B machines simulate NETMs},
    where we show that Wang~B machines simulate non-erasing
    Turing machines in linear time.  So from Sections~\ref{sec:nonerasingTM}
    and~\ref{sec:Wang B machines simulate NETMs} we get Theorem~\ref{thm:Wang efficiently universal}.
  \begin{theorem}
    \label{thm:Wang efficiently universal} 
    Let $\mathcal{M}$ be a deterministic Turing
    machine with a single binary tape that runs in time~$t$. Then there is
    a Wang~B machine $\mathcal{W}_{\mathcal{M}}$ that simulates the computation of
    $\mathcal{M}$ in time $O(t^3)$. 
  \end{theorem}

  In Section~\ref{HTMsimWH} we  give a formal description of
    Hasenjaeger's Turing machine and, for the sake of completeness, we show that
    Hasenjaeger's machine simulates Wang~B machines in linear time. So from Theorem~\ref{thm:Wang efficiently universal} and
    Section~\ref{HTMsimWH}   we
    get that Hasenjaeger's machine is an efficient  polynomial time simulator of Turing machines:

  \begin{theorem}
    \label{thm:Hasenjaegers_UTM_sim_TMs}
    Let $\TM$ be a deterministic Turing machine with a single binary tape that
    computes in time~$t$. Hasenjaeger's universal Turing machine simulates
    the computation of $\TM$ in time $O(t^3)$.
  \end{theorem}
 
     In Section~\ref{sec:Hooper} we apply Theorem~\ref{thm:Wang efficiently universal}
    to show that a small universal Turing machine of
    Hooper's~\cite{Hooper1963,Hooper1969} is efficiently universal by showing
    that it simulates Turing machines (via Wang~B machines) in polynomial
    time.

  For the remainder of this section we discuss program-size in small universal
  Turing machines.  Hasenjaeger's machine has 4 states and 2 symbols, making
  it a remarkably small universal program. However, it uses 3 non-erasable
  tapes, and so making direct comparisons with other  Turing machine models
  that have small universal programs (but have more or less tapes, tape
  dimensions, etc.)\   is not a straightforward matter. The standard model in
  the small universal Turing machine world consists of a single one
  dimensional tape with one tape head, a deterministic program, and the usual
  notion of a blank symbol~\cite{neary2012complexity}. Other more general
  models use larger numbers of tapes, higher tape dimensions, infinitely
  repeated blank words instead of a repeated blank symbol, and so on, and
  these more general models often have smaller universal programs. In the
  absence of formal tools, namely tight program-size overheads for simulations
  between these models, comparisons between them is at best challenging.
  Glaschick is the most recent author to propose a formula to compare such
  models~\cite{GlaschickINI2012}.

  As an example of the difficulty of comparing different Turing machine
  models, consider one of  Priese's~\cite{Priese1979}  universal machines.
  Priese's universal machine has 2 states, 2 symbols,  a single 2-dimensional
  tape with 2 tape heads, and uses an unconventional technique for ending its computation\footnote{Priese's machine does not end its computation using the standard method of halting on a state-symbol pair that has no transition rule: instead there is a choice, via the initial input encoding, of ending a
computation either by entering a sequence of 6 repeating configurations or by
halting when an attempt is made to move off the edge of the 2D tape.
}. However, for standard 2-state, 2-symbol machines it is known
  that no universal machines exist as their halting problem is
  decidable~\cite{Kudlek1996,Pavlotskaya1973}. So, by generalising aspects of
  the model, Priese found a  machine model whose  smallest universal programs have \emph{strictly} less 
  states and symbols than those of  the standard model. 
  Returning our attention to
  Hasenjaeger's model, we note that while his machine has 3 tapes, the size of
  his program is still impressive when one considers that 2 tapes are read-only and the work tape is non-erasing.
  
For more recent results on small non-erasing Turing machines one can look to the work of Margenstern~\cite{Margenstern1993,Margenstern1997} where he constructs small non-erasing single-tape machines and gives boundaries between universality and non-universality for various parameters in the model. More on the topic of small universal Turing machines can be found in the surveys~\cite{Margenstern2000,neary2012complexity}.
   
  \section{Non-erasing Turing machines simulate Turing machines in time $O(t^3)$}
  \label{sec:nonerasingTM}

  \begin{definition}[Binary Turing machine] \label{def:binTM}
    A binary Turing machine is a tuple $M=(Q,\{0,1\},f,q_{0},q_{|Q|-1})$. Here
    $Q$ and $\{0,1\}$ are the finite sets of states and tape symbols
    respectively. 0 is the blank symbol, $q_{0}\in Q$ is the start state, and
    $q_{|Q|-1} \in Q$ is the halt state. The transition function $f$ is of the form  $f :
    Q\times\{0,1\}\rightarrow \{0,1\}\times\{L,R\}\times Q$ and is 
    undefined on  $\{ q_{|Q|-1}  \} \times \{ 0,1\}$.
  \end{definition}

  We write $f$ as a list of transition rules.
  Each transition rule is a quintuple $(q_{i},x_1,x_2,D,q_{j})$ with
    initial state $q_{i} \in Q$,
    read symbol $x_1 \in \{ 0,1 \}$,
    write symbol $x_2 \in \{ 0,1 \}$,
    move direction $D \in \{ L,R\}$,
    and next state~$q_{j} \in Q$.

  \begin{definition}[Non-erasing Turing machine]
    A non-erasing Turing machine is a binary Turing machine where there are no
    transition rules that overwrite~1 with~0, that is, there is no
    transition rule of the form $(q_{j},1,0,D,q_{k})$, where $q_{j}, q_{k} \in Q$
    and $D \in \{L,R\}$. 
  \end{definition}

  \begin{lemma}{(Zykin~\cite{Zykin1963})}
    \label{lem:Non-ersasing-TM}
    Let $\TM$ be a deterministic single-tape binary Turing machine that runs in time $t$.
    Then there is a deterministic non-erasing Turing machine $\NETM_\TM$ that
    simulates the computation of $\TM$ in time $O(t^3)$. 
  \end{lemma}  
  \begin{proof}
    We give a brief overview of how $\TM$ is simulated by a deterministic non-erasing
       Turing machine $\NETM_\TM$ with a single tape in time $O(t^3)$.
    An arbitrary tape of $\TM$ is encoded for $\NETM_\TM$ as follows.
      Each symbol on the tape of $\TM$ is encoded as three contiguous symbols on the tape of $\NETM_\TM$.
      The two rightmost symbols of each triple encode 0 and 1 as 10 and 01 respectively.
      The leftmost symbol of the triple is 1 if and only if $\NETM_\TM$ is simulating that $\TM$'s tape head is currently reading the symbol encoded by the pair immediately to its right. 
      To simulate a timestep of $\TM$, $\NETM_\TM$ simply makes a new copy of the encoded tape of $\TM$ (to the right of the original), by scanning over and back repeatedly.
      During the copying process the encoded tape contents are appropriately modified to simulate the transition rule of $\TM$.
      This involves simulating the tape head movement of $\TM$ by copying the 1 that encodes the tape head position of $\TM$ to the left of the pair of symbols encoding the new read symbol. 
      If we are simulating a rule where $\TM$ changes a bit under its tape head, then 
       the encoded read symbol (i.e.\ the triple) is appropriately changed by $\NETM_\TM$  as it is being copied.
      The state-changes of $\TM$ can be simulated by state-changes of $\NETM_\TM$ in a straight-forward manner.

    Since $\TM$ runs in time $t$, it uses at most $t$ tape cells.
    Thus, $\NETM_\TM$ takes $O(t^2)$ steps when copying the encoding of an arbitrary configuration of $\TM$ to simulate a single step of $\TM$.
    So $t$ steps of $\TM$ are simulated by $\NETM_\TM$ in time $O(t^3)$.
  \end{proof}

  \section{Wang~B machines}
    \label{sec:Wang B machines simulate NETMs}
    A Wang~B machine is a computing machine with a single non-erasing
    bi-infinite tape that has a binary alphabet~\cite{Wang1957}. Unlike a
    Turing machine, which  performs three operations in a single timestep
    (write a 1 to its tape, move its tape head, and move program control to a
    arbitrary location in its program), a Wang~B machine can  perform only one
    operation at each timestep. Also,  in a Turing machine, control flow can
    jump to an arbitrary program location when reading~0 or~1, but a Wang
    B machine performs a control flow jump only  upon reading~1.
  
  \begin{definition}[Wang~B machine]
    \label{def:Wang B machine}
    A Wang~B machine is a finite list of instructions $\mathcal{W}=I_0,I_1,I_2,\ldots,I_{n-1}$ where each instruction is of one of the following four forms:\\
      \begin{tabular}{cl}
        $L$ &: move tape head left,\\
        $R$ &: move tape head right,\\
        $M$ &: mark the current tape cell by writing the symbol 1,\\
        $J(x)$&: if the current cell contains the symbol 1 then jump to instruction $I_x$,\\
              &\hspace{1em} otherwise move to the next instruction \\ 
      \end{tabular}
  \end{definition}
  Instructions are executed by the machine one at a time, with the 
  computation starting at instruction $I_0$. A left move or right move
  instruction ($I_k\in\{L,R\}$) moves the head one cell to the left or right
  on the tape. A mark instruction ($I_k=M$)  marks the tape: if a cell is 0
  (unmarked)  it becomes 1 (marked), otherwise if a cell is~$1$ it stays
  as~$1$.  For a jump instruction, $I_k=J(x)$, where $0 \leq x \leq n -1$, if
  the current tape cell is 1 then the machine jumps to instruction~$I_x$.
  Alternatively, when $I_k=J(x)$ and the current cell is 0 the machine will
  either move to the next instruction $I_{k+1}$ if $k<n-1$, or it will halt if
  $k=n-1$.  After each move or mark instruction $I_k$, the machine either
  moves to the next instruction $I_{k+1}$ if $k<n-1$, or halts if $k=n-1$.

    \subsection{Wang's B machines simulate non-erasing Turing machines in linear time}

  \begin{theorem}
    \label{thm:Wang B simulates NETM}
    Let $\NETM$ be a deterministic non-erasing Turing
    machine with a single binary tape that runs in time~$t$. Then there is
    a Wang~B machine $\mathcal{W}_\NETM$ that simulates the computation of
    $\NETM$ in time $O(t)$. 
  \end{theorem}
  
  \begin{proof}
    We begin by giving the program for the Wang~B machine $\WBM_\NETM$
      followed by the encoding it uses to simulate $\NETM$.
    We then show that $\WBM_\NETM$ simulates each  transition rule in $\NETM$ in constant time, 
      and so simulates the computation of $\NETM$ in time $O(t)$. 

\subsubsection{Encoding}
    Let $\tne{TR_{q_i,\sigma_1}}$ denote a sequence of Wang~B machine instructions
      that encode the transition rule
        $TR_{q_i,\sigma_1}= (q_i,\sigma_1,\sigma_2,D,q_j)$ from $\NETM$ 
          where  $q_i, q_j \in Q$, 
            $\sigma_1, \sigma_2 \in \{0,1\}$ and
            $D \in \{R,L\}$.
    The sequence of instructions for~$\WBM_\NETM$~is 
    \begin{align}
      \label{eq:W_N}
       \WBM_\NETM = & \; R,J(8),\tne{TR_{q_0,0}},\tne{TR_{q_0,1}},\notag \\
        &\hspace{0.5cm}R,J(21),\tne{TR_{q_1,0}},\tne{TR_{q_1,1}},\notag \\
          &\hspace{3cm} \vdots \notag\\
            &\hspace{1cm} R,J(13i+8),\tne{TR_{q_i,0}},\tne{TR_{q_i,1}}, \\
              &\hspace{3.5cm} \vdots \notag\\ 
                &\hspace{1.5cm} R,J(13(|Q|-2)+8),\tne{TR_{q_{|Q|-2},0}}\tne{TR_{q_{|Q|-2},1}},M \notag
    \end{align} 
    where $|Q|$ is the number of states in $\NETM$,
      and $\tne{TR_{q_i,0}}$ and $\tne{TR_{q_i,1}}$ are the instruction sequences given by 
        Equations~\eqref{eq:instructions sequence for read 0 transition rule} 
        and~\eqref{eq:instructions sequence for read 1 transition rule}. 
      
    We now define Equations~\eqref{eq:instructions sequence for read 0 transition rule}
        and~\eqref{eq:instructions sequence for read 1 transition rule}
      which give the sequence of instructions used to simulate each transition rule.
    \begin{equation}\label{eq:instructions sequence for read 0 transition rule}
        \tne{TR_{q_i,0}} =
          \begin{cases} 
            R,M,M,M,M,J(13j) & \text{if\quad} TR_{q_i,0}=(q_i,0,0,R,q_j) \\
            L,L,L,M,M,J(13j) & \text{if\quad} TR_{q_i,0}=(q_i,0,0,L,q_j) \\
            M,R,M,M,M,J(13j) & \text{if\quad} TR_{q_i,0}=(q_i,0,1,R,q_j )\\
            M,L,L,L,M,J(13j) & \text{if\quad} TR_{q_i,0}=(q_i,0,1,L,q_j) 
          \end{cases}
      \end{equation}
      \begin{equation}\label{eq:instructions sequence for read 1 transition rule}
        \tne{TR_{q_i,1}} =
          \begin{cases} 
            R,M,M,M,J(13j) & \text{if\quad} TR_{q_i,1}=(q_i,1,1,R,q_j) \\
            L,L,L,M,J(13j) & \text{if\quad} TR_{q_i,1}=(q_i,1,1,L,q_j )
          \end{cases}
      \end{equation}   
      (There are only two cases for $\tne{TR_{q_i,1}}$ as non-erasing machines never overwrite~1 with~0.)
      
    We encode the symbols 0 and 1 of $\NETM$ as $\tne{0}=10$ and $\tne{1}=11$ respectively.
    An arbitrary configuration of $\NETM$ is given by 
    \begin{equation}\label{eq:non-erasing Turing machine tape}
      q_i,\qquad w_{0}\,w_{1}\,\ldots\,w_{j-1}\,\underline{w_j}\,w_{j+1}\, \ldots\,w_{n-1}
    \end{equation}
      where $q_i$ is the current state, $w_{0}\ldots w_{n-1}$ is the tape contents, $w_k \in \{ 0,1\}$ and the tape head position is given by an underline. The configuration in Equation~\eqref{eq:non-erasing Turing machine tape} is encoded as the B machine tape 
    \begin{equation}
      \label{eq:encoded non-erasing Turing machine tape}
      I_{13i},\qquad\tne{w_{0}}\,\tne{w_{1}}\,\ldots\,\tne{w_{j-1}}\,\underline{\bm{w_{j_1}}}\bm{w_{j_2}}\,\tne{w_{j+1}} \ldots\,\tne{w_{n-1}} 
    \end{equation}
    where $\tne{w_k}\in \{\tne{0},\tne{1}\}$, the encoded read symbol $w_{j_1}w_{j_2}=\tne{w_j}$ is given in bold, and $I_{13i}$ is the next instruction to be executed and encodes that $\NETM$ is in state $q_i$. Note that $I_{13i}$ is the first instruction in the sequence $I_{13i},\ldots, I_{13i+12}=R,J(13i+8),\tne{TR_{q_i,0}},\tne{TR_{q_i,1}}$ that encodes the pair of transition rules for state~$q_i$.
 
The infinite number of blank tape cells of $\NETM$ each contain the symbol $0$, as do the blank tape cells of $\WBM_\NETM$.  Note that, during the simulation, $\WBM_\NETM$ may need to simulate the situation where the tape head of $\NETM$ moves to a blank tape cell. In this case, as described below, the simulator $\WBM_\NETM$ will move to the relevant blank portion of its own tape and convert the symbol pair $00$ to $10 = \tne{0}$.

    \subsubsection{Simulating transition rules}\label{sec:simulatingTRs} 
     At the start of each simulated timestep of machine $\NETM$, our Wang machine $\WBM_\NETM$ has a configuration of the form given in Equation~\eqref{eq:encoded non-erasing Turing machine tape}. Each simulated timestep begins with $\WBM_\NETM$ choosing which transition rule to simulate by reading the encoded read symbol and then choosing which sequence ($\tne{TR_{q_i,0}}$ or $\tne{TR_{q_i,1}}$) to execute. 
       
    From Equation~\eqref{eq:encoded non-erasing Turing machine tape}, each simulated timestep begins with the tape head over the leftmost symbol of the encoded read symbol ($\tne{0}=10$ or $\tne{1}=11$). So, immediately after we execute the first instruction (which is $I_{13i}=R$) the tape head is over the rightmost symbol of $\tne{0}$ or $\tne{1}$ and the program control is at instruction $I_{13i+1} = J(13i+8)$. If we are reading $\tne{0} = 10$, then the rightmost symbol is a 0 and so no jump occurs on $J(13i+8)$. This means that control moves to instruction $I_{13i+2}$, the leftmost instruction in $\tne{TR_{q_i,0}}$. Alternatively, if we are reading $\tne{1} = 11$, then the rightmost symbol is a 1 and $J(13i+8)$ will jump to instruction  $I_{13i+8}$, sending control to the leftmost instruction of $\tne{TR_{q_i,1}}$.
    (To see this, use Equation~\eqref{eq:W_N} to count the number of
    instructions that precede $\tne{TR_{q_i,1}}$, which gives $13i+8$,
    specifically 13 instructions for each state $q_j$ where $j<i$ and a
    further 8 instructions for the sequence $R,J(13i+8),\tne{TR_{q_i,0}}$.) 
        
    We now explain how the sequences in Equations~\eqref{eq:instructions sequence for read 0 transition rule} and~\eqref{eq:instructions sequence for read 1 transition rule} simulate the transition rules of $\NETM$. 

      \paragraph{Case 1. Read symbol of $\NETM$ is 1}
      As described at the beginning of Section~\ref{sec:simulatingTRs}, the simulation of each timestep begins with the execution of  $R,J(13i+8)$. When the read symbol of $\NETM$ is 1, and the pair of instructions $R,J(13i+8)$ have  executed, we have the following tape contents for $\WBM_\NETM$
      \begin{equation}\label{eq:initial configuration transition rule simulation read 1}
        \tne{w_{0}}\,\tne{w_{1}}\,\ldots\,\tne{w_{j-2}}\,10\;\bm{1}\underline{\bm{1}}\;\tne{w_{j+1}} \ldots\,\tne{w_{n-1}} 
      \end{equation}
      (For illustration purposes, we assume that in $\NETM$ the symbol to the left of the read symbol is a $0$, which is encoded as $\tne{0}=10$ in Equation~\eqref{eq:initial configuration transition rule simulation read 1}.)
      
         As described at the beginning of Section~\ref{sec:simulatingTRs}, when the read symbol of $\NETM$ is~$1$ the execution of $R,J(13i+8)$ is followed by the execution of the sequence $\tne{TR_{q_i,1}}$. If we are simulating $(q_i,1,1,L,q_j)$, then from Equation~\eqref{eq:instructions sequence for read 1 transition rule} the instruction  sequence $\tne{TR_{q_i,1}}=L,L,L,M,J(13j)$ is applied to the tape in Equation~\eqref{eq:initial configuration transition rule simulation read 1} to give
  \begin{equation}\label{eq:final configuration transition rule simulation read 1}
  \tne{w_{0}}\,\tne{w_{1}}\,\ldots\,\tne{w_{j-2}}\,\underline{\bm{1}}\bm{0}\;11\;\tne{w_{j+1}} \ldots\,\tne{w_{n-1}} 
  \end{equation}
  The tape in Equation~\eqref{eq:final configuration transition rule simulation read 1}
  is of the form  in Equation~\eqref{eq:encoded non-erasing Turing machine tape},
  hence  the tape configuration is  ready for the simulation of the next
  timestep. The jump instruction~$J(13j)$ sent the program control
  of~$\WBM_\NETM$ to the first instruction of the sequence of instructions that
  encodes state $q_j$. This is verified by counting the number of instructions
  to the left of $R,J(13j+8),\tne{TR_{q_j,0}},\tne{TR_{q_j,1}}$ using the same
  technique as above.  So, the simulation of the transition rule
  $(q_i,1,1,L,q_j)$ is complete.

To generalise this example to all possible cases for simulating a rule of the form $(q_i,1,1,L,q_j)$ we need only consider the encoded symbol (from $\NETM$) immediately to the left of the encoded  read symbol (in our analysis $i,j$ are already arbitrary). If the encoded symbol to the left of the tape head in Equation~\eqref{eq:initial configuration transition rule simulation read 1} was $\tne{1}=11$ instead of $\tne{0}=10$, then it is verified in the same straightforward manner. If we are simulating the situation where $\NETM$ is at the left end of its tape (the tape is blank to the left: all $0$s) and so contains the pair $00$ immediately to the left of the encoded read symbol in Equation~\eqref{eq:initial configuration transition rule simulation read 1}. This $00$ pair is changed to $\tne{0}=10$ by the $M$ instruction that immediately proceeds the $J(13j)$ instruction, correctly providing the symbol pair $10$ that encodes the $0$ as $\tne{0} = 10$ for the next simulated timestep. Also, the 1 printed by this $M$ instruction allows instruction $J(13j)$ to jump the program control to the encoding of the next state $q_j$.

  The case of simulating $(q_i,1,1,R,q_j)$ is verified by applying the sequence $\tne{TR_{q_i,1}}=R,M,M,M,J(13j)$ from Equation~\eqref{eq:instructions sequence for read 1 transition rule} to the tape in Equation~\eqref{eq:initial configuration transition rule simulation read 1}. This analysis is similar to the previous example and so we omit the details. 
  We simply note that after the first of the three $M$ instructions is executed, the tape cell will always contain a 1 and so the second and third $M$ instructions do not change the tape. (These extra $M$ instructions are used for padding so that each encoded state has exactly 13 instructions.)

  \paragraph{Case 2. Read symbol of $\NETM$ is 0} 
      As described at the beginning of Section~\ref{sec:simulatingTRs}, the
      simulation of each timestep begins with the execution of  $R,J(13i+8)$.
      When the read symbol of $\NETM$ is 0, and the pair of instructions
      $R,J(13i+8)$ have been executed, we have the following tape contents for
      $\WBM_\NETM$.
  \begin{equation}\label{eq:initial configuration transition rule simulation read 0}
  \tne{w_{0}}\,\tne{w_{1}}\,\ldots\,\tne{w_{j-1}}\;\bm{1}\underline{\bm{0}}\;11\;\tne{w_{j+2}} \ldots\,\tne{w_{n-1}} 
  \end{equation}
  (For illustration purposes, we assume that in $\NETM$ the symbol to the right of the read symbol is a $1$, which is encoded as $\tne{1}=11$ as in Equation~\eqref{eq:initial configuration transition rule simulation read 0}.)        
        
        As described at the beginning of Section~\ref{sec:simulatingTRs}, when the read symbol of $\NETM$  is $0$ the execution of $R,J(13i+8)$ is followed by the execution of the sequence $\tne{TR_{q_i,0}}$. If we are simulating $(q_i,0,1,R,q_j)$, then from Equation~\eqref{eq:instructions sequence for read 0 transition rule} the sequence $\tne{TR_{q_i,0}}=M,R,M,M,M,J(13j)$ is applied to the tape in Equation~\eqref{eq:initial configuration transition rule simulation read 0} to give
  \begin{equation}\label{eq:final configuration transition rule simulation read 0}
  \tne{w_{0}}\,\tne{w_{1}}\,\ldots\,\tne{w_{j-1}}\;11\;\underline{\bm{1}}\bm{1}\;\tne{w_{j+2}} \ldots\,\tne{w_{n-1}}
  \end{equation}
  The first $M$ instruction changed  $\tne{0}=10$ to $\tne{1}=11$ simulating the printing of the write symbol by $\NETM$. 
  The tape in Equation~\eqref{eq:final configuration transition rule simulation read 0} is of the form found in Equation~\eqref{eq:encoded non-erasing Turing machine tape} and is ready for the simulation of the next transition rule to begin. The jump instruction $J(13j)$ sends the program control of $\WBM_\NETM$ to the instruction sequence of the program that encodes state $q_j$. This is verified using the same technique as in the previous case. So, the simulation of $(q_i,0,1,R,q_j)$ is complete.

  To generalise this example to all possible cases for simulating a rule of the form $(q_i,0,1,R,q_j)$, we need only consider the encoded symbol (from $\NETM$) immediately to the right of the encoded read symbol (in our analysis $i,j$ are already arbitrary). If the encoded symbol to the right of the tape head in Equation~\eqref{eq:initial configuration transition rule simulation read 0} was $\tne{1}=10$ instead of $\tne{1}=11$, then it is verified in the same straightforward manner.  If we are simulating the situation where $\NETM$ is at the right end of its tape  (the tape is blank to the right: all $0$s) and so contains the pair $00$ immediately to the right of the encoded read symbol in Equation~\eqref{eq:initial configuration transition rule simulation read 0}. This $00$ pair is changed to $\tne{0}=10$ by the second $M$ in the sequence  $M,R,M,M,M,J(13j)$ which provides the symbol pair $10 =  \tne{0}$ that correctly encodes a $0$ for the next simulated timestep.  Also, the 1 printed by this $M$ instruction allows instruction $J(13j)$ to jump the program control to the encoding of the next state $q_j$. As with the previous case, the extra~$M$ instructions are added for padding.
    
  The other cases for simulating $\NETM$ reading a $0$ are verified by applying the appropriate sequences from Equation~\eqref{eq:instructions sequence for read 0 transition rule} to the tape in Equation~\eqref{eq:initial configuration transition rule simulation read 0}. The details are similar to the previous example and are omitted. 
  
  \subsubsection{Halting and time complexity.}
        When $\NETM$ enters its halt state, defined to be state $q_{|Q|-1}$ in Definition~\ref{def:binTM}, then~$\WBM_\NETM$ executes the jump instruction $J(13(|Q|-1))$ and jumps to the rightmost instruction in Equation~\eqref{eq:W_N}, an $M$ instruction. Note that in order to jump to this $M$ instruction we must have read a 1 on the tape, and so this $M$ instruction does not change the tape. 
        After executing this $M$ instruction, $\WBM_\NETM$ is at the end of its list of instructions and so it halts. 
        
        From Equations~\eqref{eq:W_N},~\eqref{eq:instructions sequence for read 0 transition rule} and~\eqref{eq:instructions sequence for read 1 transition rule}, exactly 13 instructions are used to encode the pair of transition rules for each state $q_i$ of $\NETM$. Furthermore, from the above algorithm, the simulation of one of these transition rules involves the execution of at most 8 instructions (at most 8 timesteps).  Thus $\WBM_\NETM$  simulates~$t$ steps of an arbitrary non-erasing Turing machine~$\NETM$ in time $O(t)$. 
    \end{proof}

  \section{Hasenjaeger's electromechanical universal Turing machine}
  \label{HTMsimWH}  
 We begin this section by briefly describing the electromechanical device constructed by Hasenjaeger~\cite{Hasenjaeger1987}, which implements a multi-tape Turing machine.   As mentioned in Section~\ref{sec:intro}, Glaschick~\cite{GlaschickINI2012} reverse engineered the physical wiring of Hasenjaeger's electromechanical machine to find the Turing machine program left by the previous programmer, presumably Hasenjaeger, and with the help of Hasenjaeger's notes saw that it simulates Wang~B machines. For completeness we include a proof that this program (wiring) for Hasenjaeger's machine simulates Wang~B machines in linear time.\footnote{Note that the machine  can be re-programmed by re-wiring.}
   
    First we briefly describe Hasenjaeger's electromechanical machine, which is shown in Figure~\ref{fig:Hasenjaeger's UTM}.  
    
    \begin{itemize}
      \item The \emph{control unit} is constructed from 16 electromechanical relays which 
        encode the \emph{main program} (also called the state table) of the Hasenjaeger machine.  
        This unit is limited to 4 states and operates on three tapes. 
      \item The \emph{program tape} ($P$) is a device consisting of 
        20 switches,  18 of which are connected, and together represent  
          a cyclic, bi-directional, read-only binary tape with 18 cells. (This short tape can be used to store a simulated program.)   
      \item The \emph{counter tape} ($C$) consists of two selector switches that represent a bi-directional,
        cyclic, read-only tape with 18 cells. 
        It represents a tape where all cells contain a 1 except for a single cell that contains a 0. 
      \item The \emph{work tape} ($W$) is a bi-directional non-erasable ``infinite'' tape.%
        \footnote{It is expected that the recent precipitous decline in the production of 
        35mm film and 
        paper punch tape will negatively impact the computing power of Hasenjaeger's machine.}
    \end{itemize}
    
    Hasenjaeger's electromechanical device, as wired, is an instance of a Turing machine. 
    However, exactly what kind of Turing machine is a matter of opinion:
      there are a number of reasonable 
      generalizations of this single device (machine instance) 
      to get a general model of computation, here we give one.
    Formally, we write the tuple $ (Q, f, q_s)$ to denote an instance of a
      three-tape Turing machine of the following form. 
    The three  tapes are bi-directional and are denoted $P$, $C$ and $W$.
    Each tape has alphabet~$\{0,1\}$ and blank symbol~$0$. 
    Tapes $P$ and $C$ are read-only, while $W$ is non-erasing 
      (i.e.\ 1s can not be overwritten with 0s).
    To give an instance of such a machine, 
      we would assign values to the tuple $(Q, f, q_s)$,
        where $Q$ is a set of states,
        $f$ is a transition function (or transition  table), of the form 
      $f : Q \times 
          \{ 0,1 \} \times
          \{ 0,1 \} \times 
          \{ 0,1 \}
        \rightarrow
          \{ L,R,\nopTape \} \times
          \{ L,R,\nopTape \} \times 
          \{ L,R,\nopTape,1\} \times
          Q$,
      and $q_s \in Q$ is  the start state.

    The machine works as follows.
    In state $q \in Q$,
      the machine reads a symbol from each of the tapes $P$, $C$, and $W$ and,
        as dictated by $f$, for each tape does one of three things: 
          move left ($\lftTape$),
          move right ($\rgtTape$),
          do nothing ($\nopTape$).
      However, for the tape $W$ 
        it has an additional fourth option of 
          marking ($\mrkTape$) the tape cell with the symbol~$1$.

    Now we formally specify Hasenjaeger's machine  $\HM = (Q, f, q_s)$  as an instance of the above model. 
    $\HM$ has four states $Q =\{ q_1,q_2,q_3,q_4 \}$
      and the start state is $q_s = q_1$.
    The function $f$ is given as a list of transition rules in Table~\ref{fig:Hrules}. 
    This table of behaviour is derived from the wiring of the 
      electromagnetic relays of Hasenjaeger's  device.
    
  \begin{table}[ht] 
    \begin{equation*}
      \begin{array}{c|cccc|cccc}
      \text{Rule } &&&&&&&& \vspace{-0.9ex}\\
      \text{number } & Q & P & C & W & P  & C   & W        & Q'\\
      \hline
      1 & q_1 & 1 & * & 0 & \rgtTape & \nopTape & 1        & q_1 \\
      2 & q_1 & 1 & * & 1 & \rgtTape & \nopTape & \nopTape & q_1 \\
      3 & q_1 & 0 & * & * & \rgtTape & \nopTape & \nopTape & q_2 \\
      \hline                                                     
      4 & q_2 & 1 & 0 & * & \rgtTape & \nopTape & \rgtTape & q_1 \\
      5 & q_2 & 0 & 0 & * & \rgtTape & \rgtTape & \nopTape & q_2 \\
      6 & q_2 & 1 & 1 & * & \rgtTape & \lftTape & \lftTape & q_1 \\
      7 & q_2 & 0 & 1 & * & \rgtTape & \lftTape & \nopTape & q_3 \\
      \hline                                                     
      8 & q_3 & 0 & * & 0 & \rgtTape & \nopTape & \nopTape & q_3 \\ 
      9 & q_3 & 1 & * & 0 & \rgtTape & \nopTape & \nopTape & q_1 \\
     10 & q_3 & 0 & * & 1 & \rgtTape & \rgtTape & \nopTape & q_3 \\
     11 & q_3 & 1 & * & 1 & \lftTape & \rgtTape & \nopTape & q_4 \\ 
      \hline                                                     
     12 & q_4 & 0 & 1 & * & \lftTape & \nopTape & \nopTape & q_4 \\
     13 & q_4 & 1 & 1 & * & \lftTape & \lftTape & \nopTape & q_4 \\
     14 & q_4 & * & 0 & * & \rgtTape & \nopTape & \nopTape & q_1 \\
      \end{array}
    \end{equation*}
    \caption{The program $f$ for Hasenjaeger's universal machine $\HM$ that simulates Wang's B machines.
    The $*$ symbol denotes that the read symbol can be 0 or 1.
    The  rule numbers on the left are not part of the program.}\label{fig:Hrules}
  \end{table}
    
  \begin{lemma}
    Let $\WBM$ be a Wang~B machine that runs in time $t$. 
    The multitape Turing machine $\HM$, defined above, simulates the computation of $\WBM$ in time $O(t)$.
  \end{lemma}
  \begin{proof}
      We begin by giving the encoding used by the program $\WBM$, followed by
      a description of how the program simulates each of the four Wang~B
      machine instructions as well as halting. We finish by giving the time
      analysis for this simulation.
    \paragraph{Encoding}
      The four Wang~B machine instructions $M$, $R$, $L$ and $J(x)$ are encoded as binary words as follows:
      $\tne{M} = 1$, $\tne{R} = 01$, $\tne{L} = 001$, and $\tne{J(x)} = 0000^y1$ 
      (the value $y\in  \{ 0,1,2,\ldots \}$ will be defined later).
      The Wang~B machine program $\WBM = I_0,I_1,\ldots , I_{n-1}$ 
      is encoded as a single binary word via Equation~\eqref{eq:Wang machine to Hasenjaeger UTM}. 
 
      \begin{equation}\label{eq:Wang machine to Hasenjaeger UTM}
        \tne{\WBM} = \tne{I_0}\tne{I_1}\tne{I_2}\ldots\tne{I_{n-2}}\tne{I_{n-1}}\tne{J(n)}
      \end{equation}
     
      The word $\tne{\WBM} \in \{ 0,1\}^*$ is placed on $\HM$'s circular
      program tape~$P$.  The $C$ tape is  defined to have length $n+2$, with
      $n+1$ of these cells containing the symbol~$1$, and the single remaining
      cell containing the symbol $0$. The $W$ tape has the same tape contents
      as that of  the Wang~B machine it simulates. At the beginning of a
      simulated computation step the tape head of $P$ is over the leftmost
      symbol of the encoded instruction it is simulating, $C$'s tape head is
      over its single $0$ symbol, and the tape head of $W$ has the same
      location as the tape head of the Wang~B machine it simulates.    
    
    To help simplify our explanation, we give partial configurations for $\HM$
    where we display a small part of each tape surrounding the tape head. For
    example, the following  configuration occurs at the beginning of a
    simulated computation step
	  \begin{xalignat*}{4}
	    q_1 &  
          & P = \ldots 1\, \underline{0}01\ldots &
          & C = \ldots 1\underline{0}1\ldots &
          & W = \ldots 1\underline{0}0\ldots 
    \end{xalignat*}
    Here, $\HM$'s current state is $q_1$  and the position of each of the
    three tape heads is given by an underline.  Also, in the above example
    the tape head of $P$ is over the leftmost symbol of an encoded left move
    instruction $\tne{L}=001$, and the $C$ tape head is at cell~$C_0$. 
    
   \paragraph{Simulate $M$ instruction} 
       The Wang~B machine $M$ instruction is encoded as $\tne{M} = 1$ on the $P$ tape. 
       If the tape head of  $W$ is reading a 0 then we have a configuration of the form
      	  \begin{xalignat*}{4}
	    q_1&  & P = \ldots 1\, \underline{1} \, 001 \ldots & & C = \ldots  1\underline{0}1\ldots & &W = \ldots  1\underline{0}0\ldots 
          \end{xalignat*}
       (For the purposes of explanation, we have assumed that there is an encoded $L$ instruction, given by $\tne{L} = 001$, to the right of $\tne{M} =1$ on the $P$ tape.) Rule 1 from Table~\ref{fig:Hrules} is applied to the above configuration to give
      	  \begin{xalignat*}{4}\label{configmovedone}
	    q_1&  & P = \ldots 1\, 1 \, \underline{0}01 \ldots & & C = \ldots  1\underline{0}1\ldots & &W = \ldots  1\underline{1}0\ldots 
          \end{xalignat*}
          The $M$ instruction was simulated by printing a 1 to the $W$ tape. Note that the tape head on the $P$ tape has moved to the leftmost symbol of the next encoded instruction ($\tne{L} = 001$), and the current state of $\mathcal{H}$ is once again $q_1$. So the simulation of the $M$ instruction is complete and $\mathcal{H}$ is configured to begin simulation of the next Wang machine instruction. 
          
          In the case where the tape head of $W$ is reading a $1$, we simulate the $M$ instruction by executing rule 2 from Table~\ref{fig:Hrules}. This is very similar to the previous case above and so we omit the detail.

    \paragraph{Simulate $R$ instruction} The Wang~B machine \emph{right move} instruction is encoded as $\tne{R} = 01$ on the $P$ tape. 
   If the tape head of $W$ is reading a 0 then we have a configuration of the form
        \begin{xalignat*}{4}
	    q_1&  & P = \ldots 1\, \underline{0}1 \, 001 \ldots & & C = \ldots  1\underline{0}1\ldots & &W = \ldots  1\underline{0}0\ldots 
          \end{xalignat*}
      Rules 3 and 4 from Table~\ref{fig:Hrules} are applied to the above configuration to give
        \begin{xalignat*}{4}
	    q_1&  & P = \ldots 1\, 01 \, \underline{0}01 \ldots & & C = \ldots  1\underline{0}1\ldots & &W = \ldots  10\underline{0}\ldots 
          \end{xalignat*}
          The tape head of $W$ was moved one place to the right to simulate the $R$ instruction. Also, the tape head on the $P$ tape has moved right 2 places to the leftmost symbol of the next encoded instruction ($\tne{L} = 001$), and the current state of $\mathcal{H}$ is once again $q_1$. So the simulation of the $R$ instruction is complete and $\mathcal{H}$ is configured to begin simulation of the next Wang machine instruction. In the case where the tape head of $W$ is reading a $1$, the computation proceeds in the same manner as above by executing rules 3 and 4 from Table~\ref{fig:Hrules}.

    \paragraph{Simulate $L$ instruction}
    The Wang~B machine \emph{left move} instruction is encoded as $\tne{L} = 001$ on the $P$ tape. 
    If the tape head of $W$ is reading a 0 then we have a configuration of the form
    \begin{xalignat*}{4}
	    q_1&  & P = \ldots 1\, \underline{0}01 \, 01 \ldots & & C = \ldots  1\underline{0}1\ldots & &W = \ldots  1\underline{0}0\ldots 
    \end{xalignat*}  
    Rules 3, 5 and 6 from Table~\ref{fig:Hrules} are applied to the above configuration to give
    \begin{xalignat*}{4}
	    q_1&  & P = \ldots 1\, 001 \, \underline{0}1 \ldots & & C = \ldots  1\underline{0}1\ldots & &W = \ldots  \underline{1}00\ldots 
    \end{xalignat*}  
     The tape head of $W$ was moved one place to the left to simulate the $L$ instruction. Also, the tape head on the $P$ tape has moved right 3 places to the leftmost symbol of the next encoded instruction ($\tne{R} = 01$), and the current state of $\mathcal{H}$ is once again $q_1$. So the simulation of the $L$ instruction is complete and $\mathcal{H}$ is configured to begin simulation of the next Wang machine instruction. In the case where the tape head of $W$ is reading a $1$, the computation proceeds in the same manner as above by executing rules 3, 5 and 6 from Table~\ref{fig:Hrules}.
    
    \paragraph{Simulate $I_k=J(x)$ instruction}
      There are two cases to consider here, which are determined by the value of
      read symbol of the simulated Wang~B machine.

    \emph{Case 1. Wang~B machine's read symbol is $0$.} 
    In this case, $\HM$ simulates program control for $\WBM$ moving from instruction~$I_k$ to instruction~$I_{k+1}$.
    This is simulated by moving the tape head to the leftmost symbol of $\tne{I_{k+1}}$.
    Instruction $I_k=J(x)$ is encoded as
      $\tne{I_k} = \tne{J(x)}=\underline{0}000^y1$
      for some $y \in \{ 0,1,2, \ldots \}$ ($y$ is defined below),
      and for the purposes of explanation we assume that $I_{k+1}=L$. 
    This gives the  configuration
    \begin{xalignat*}{4}
	    q_1 &  
          & P = \ldots 1\, \underline{0}000^y1\, 001 \ldots &
          & C = \ldots \underline{0}1 \ldots & 
          & W = \ldots \underline{0}\ldots
    \end{xalignat*}
    After applying rules 3, 5 and 7 from Table~\ref{fig:Hrules} we get the following
    \begin{xalignat*}{4}
      q_3& 
         & P = \ldots1\, 000\underline{0}0^{y-1}1\, 001 \ldots & 
         & C = \ldots \underline{0}1 \ldots & 
         & W = \ldots \underline{0} \ldots
    \end{xalignat*}
    Next, rule 8 is applied $y$ times followed by a single application of rule 9 to give
    \begin{xalignat*}{4}
	    q_1& 
         & P = \ldots 1\, 0000^{y}1\, \underline{0}01\ldots & 
         & C = \ldots \underline{0}1\ldots & 
         & W = \ldots \underline{0} \ldots
    \end{xalignat*}
    In the configuration immediately above, the simulation of $J(x)$ when the
    Wang machine read symbol is 0 is complete. Note that  $\HM$ has returned
    to state $q_1$ and the tape head of $P$ is over the leftmost symbol of the
    encoded instruction $\tne{I_{k+1}}=001$. 
          
    \emph{Case 2. Wang~B machine read symbol is 1.}
     In this case, simulating the instruction $I_k=J(x)$ 
      involves    moving the $P$ tape head to the leftmost symbol of $\tne{I_{x}}$. 

    We begin with an overview, which includes specifying the encoding of jump instructions.
    Each encoded instruction contains a single 1 symbol, and so as we move
    through the $P$ tape we can count the number of encoded instructions by
    counting the number of 1 symbols.  
    If $x \leqslant k$, then, from Equation~\eqref{eq:Wang machine to Hasenjaeger UTM}, 
    we can move from $\tne{I_k}$ to $\tne{I_x}$ on $P$ by moving 
    left until we have read the symbol 1 exactly $k-x +1$ times,
      and then moving right.
    Recall that the $P$ tape is circular, and so if $x > k$,
      using Equation~\eqref{eq:Wang machine to Hasenjaeger UTM},
      we move from $\tne{I_k}$ to $\tne{I_x}$ on $P$ by moving left until
      we have read the 1 symbol exactly  $(n+2+k-x)$ times, and then moving
      right. 
    We are now ready to give the encoding for jump instructions.
      $$\tne{I_k} = \tne{J(x)}=\underline{0}000^y1$$ where  
    \begin{equation}\label{eq:y}
         y =
          \begin{cases} 
            k - x       & \text{if\quad} x \leqslant k \\
            n+1 + k - x & \text{if\quad} x > k 
          \end{cases}
    \end{equation} 

    In the simulation, moving from $\tne{I_{k}}$ to $\tne{I_{x}}$ is done in 2 stages.
    In the first stage the word
      $\tne{J(x)}=\underline{0}000^y1$ is read and the value $y+1$ is recorded 
      by the tape head position on the $C$ tape.
    In the second stage, using the value stored on the $C$ tape, 
      the tape head of $P$ moves left until we have read the symbol $1$ exactly $y+1$ times.
    So, the tape head of $P$ finishes its scan left immediately to the left of the 1 in $\tne{I_{x-1}}$; from there it moves right two cells to the leftmost symbol of $\tne{I_{x}}$. 

    Now we give the details of how $\HM$ simulates a jump from instruction $I_{k}$ to instruction $I_{x}$.
    For the purposes of illustration we assume the instruction to the left of
    $I_{k}$ is $I_{k-1}=L$. This gives the configuration
    \begin{xalignat*}{4}
	    q_1 &  
          & P = \ldots 001\, \underline{0}000^y1\ldots &
          & C = \ldots \underline{0}1^{y+1}\ldots &
          & W = \ldots \underline{1}\ldots
    \end{xalignat*}
    First, rules 3, 5 and 7 from Table~\ref{fig:Hrules} are applied,
      and then rule 10 is applied $y$ times, followed by a single application of rule 11, to give         
    \begin{xalignat*}{4}
	    q_4&
         & \ldots P = 001\, 0000^{y-1}\underline{0}1\ldots & 
         & C = \ldots 01^{y}\underline{1}\ldots & 
         & W = \ldots \underline{1} \ldots
    \end{xalignat*} 
    In the configuration immediately above the value $y+1$ is recorded by the
    position of the tape head of $C$, which is over $y+1^{\textrm{th}}$ symbol
    to the right of the single $0$ symbol.
    Rule 12 is applied $y+3$ times to give
    \begin{xalignat*}{4}
	    q_4& 
         & \ldots P =  00\underline{1}\, 0000^{y}1\ldots & 
         & C = \ldots 01^{y}\underline{1}\ldots & 
         & W = \ldots \underline{1} \ldots
    \end{xalignat*}
    When 1 is read on tape $P$ the value stored on tape $C$  is decremented by
    moving left once on  $C$ using rule 13. This gives
    \begin{xalignat*}{4}
      q_4&  
         & \ldots P =  0\underline{0}1\, 0000^{y}1\ldots &
         & C = \ldots 01^{y-1}\underline{1}1\ldots & 
         & W = \ldots \underline{1} \ldots
    \end{xalignat*}
    The above process of decrementing the value stored in $C$ by applying rules 12 and 13 continues until the tape head of $C$ reads a $0$, indicating that the scan left is finished (during this process Rule 13 is applied a total of $y+1$ times). At this point we have a configuration that is of one of the following two forms
	   \begin{xalignat*}{4}
	    q_4&  & P =\ldots  \underline{0}1\ldots  & & C = \ldots \underline{0}1\ldots & & W = \ldots \underline{1} \ldots\\
	    q_4&  & P =\ldots  \underline{1}1\ldots  & & C = \ldots \underline{0}1\ldots & & W = \ldots \underline{1} \ldots
	    \end{xalignat*} 
	    Rule 13 was applied $y+1$ times reading a 1 each time. 
      Rules 14 and 2 are applied to move the tape head of $P$ right twice,
      placing it over the leftmost symbol of instruction $\tne{I_{x}}$ to
      complete the simulation of $I_k=J(x)$. 

  \paragraph{Simulation of halting} 
    Recall from Section~\ref{sec:Wang B machines simulate NETMs},
    that a Wang~B machine halts when it attempts to move to 
    the non-existent instruction $I_{n}$ after executing instruction $I_{n-1}$.
    Since  $\HM$  does not have a
    distinguished halt state, it instead
    simulates halting by entering a repeating sequence of configurations. 
    Note that in Equation~\eqref{eq:Wang machine to Hasenjaeger UTM}, 
      as part of the Wang~B machine encoding,
      there is an extra instruction ($I_{n}=J(n)$) that jumps to itself.
    So when program $\mathcal{H}$ simulates a Wang~B machine that halts by 
    attempting to move to instruction $I_{n}$, the program 
    simulates the instruction $J(n)$ 
    which results in an infinite loop and signals the end of the simulation.
    (The jump instruction works as intended only if we have the
    assumption that the cell under $W$'s tape head reads 1; it is easy to
    modify any Wang~B program so that this is the case by having the program
    end with a single mark instruction, i.e.\ $I_{n-1} = M$).
    This jump works as follows.          
    From Equation~\eqref{eq:y}, a jump instruction of the form $I_n = J(n)$ is
    encoded as $\tne{I_n } =  \tne{J(n)} = 0001$.     
    This gives the configuration
    \begin{xalignat}{4}\label{eq:loop}
	    q_1 &  
          & P = \ldots \underline{0}001\ldots &
          & C = \ldots \underline{0}1\ldots &
          & W = \ldots \underline{1}\ldots
    \end{xalignat}
    From here, $\mathcal{H}$ simulates the jump instruction $I_n = J(n)$, as
    described above. In this simple case, simulating the jump instruction
    involves executing exactly 10 rules (see Table~\ref{fig:Hrules}) after
    which $\mathcal{H}$ returns to configuration~\eqref{eq:loop}. Hence we get
    an infinite loop where the tape contents are unchanged.
    
    \paragraph{Complexity analysis}
      The Wang~B machine instructions $M$, $R$, and $L$ are each simulated by
      Hasenjaeger's machine in 1, 2 and 3 timesteps, respectively. The $J(x)$
      instruction is simulated in $O(n^2)$ timesteps, where $n$ is the number
      of instructions in the Wang~B machine program. Note that we consider $n$
      to be a constant, independent of the input length. 
      Therefore, Hasenjaeger's  program~$\HM$ simulates $t$ steps of the Wang~B machine~$\mathcal{W}$ in time $O(t)$.
    \end{proof}

\section{Hooper's small universal Turing machine simulates Turing machines in polynomial time}\label{sec:Hooper}
      Hooper~\cite{Hooper1963,Hooper1969} gave a small universal Turing machine with 1 state, 2 symbols and 4 tapes. Using similar techniques to Hasenjaeger, Hooper proved his machine universal by simulating a restricted class of Wang~B machines. In Hooper's machine, a non-erasing work tape contains exactly the same contents as the tape of the Wang~B machine it simulates, and a read-only unidirectional circular  program tape stores the encoded Wang machine program. Hooper used a relative addressing technique like Hasenjaeger, but unlike Hasenjaeger, Hooper used two read-write counter tapes (instead of one read-only tape). Hooper's simulation of Wang~B machines runs in linear time, and so from Lemma~\ref{lem:Non-ersasing-TM} and by suitably  modifying the proof of Theorem~\ref{thm:Wang B simulates NETM}  we get the following result.  
    \begin{theorem}
    Let $\mathcal{M}$ be a deterministic Turing machine with a single binary tape that runs in time~$t$. Then Hooper's small universal Turing machine  with 1~state, 2~symbols, and 4~tapes~\cite{Hooper1963,Hooper1969}  simulates the computation of $\mathcal{M}$ in time~$O(t^3)$.
    \end{theorem}
    \begin{proof}
    Hooper's machine simulates Wang~B machines with the following restrictions:
    \begin{enumerate}
     \item In the program list if $I_k=J(x)$, then $I_{k+1}\in\{L,R\}$.
     \item Each jump instruction jumps to  $\{L,R\}$. 
     \item $M$ instructions are executed only on tape cells that contain 0.
    \end{enumerate}
    Our proof of Theorem~\ref{thm:Wang B simulates NETM} is easily modified to include the above restrictions. For restriction 1, we add the instruction sequence $L,R$ after each jump instruction in the program. This has no effect on the program as a move left followed by a move right has the same effect  as no move. Our proof already satisfies restriction~2, as we either jump to the beginning of the sequence encoding a state $q_i$ (that is:  $R,J(13i+8),\tne{TR_{q_i,0}}, \tne{TR_{q_i,1}}$) or we jump to the beginning of a sequence of the form $\tne{TR_{q_i,1}}$ (given in Equation~\eqref{eq:instructions sequence for read 1 transition rule}). To satisfy restriction~3, each $I_k=M$ instruction is replaced with the sequence $J(k+4),R,L,M,R,L$. The $J(k+4)$ will jump over the $M$ instruction if the cell already contains a 1, and the extra $R,L$ instructions are introduced to satisfy restrictions 1 and 2. 
    
    In addition to the above changes, we wish to maintain the property from the proof of Theorem~\ref{thm:Wang B simulates NETM} that the number of instructions used to encode each Turing machine state is the same for all states. Recall from Theorem~\ref{thm:Wang B simulates NETM} that each state $q_i$ is encoded as the sequence of 13 instructions $R,J(13i+8),\tne{TR_{q_i,0}},\tne{TR_{q_i,1}}$. This sequence has 3 jump instructions and to satisfy restriction 1 we added the extra instruction pair $L,R$ for each jump. For restriction 3, we replaced each $M$ instruction with $J(k+4),R,L,M,R,L$. In Equation~\eqref{eq:instructions sequence for read 1 transition rule} this gives an extra 15 instructions for the case $(q_i,1,1,R,q_j)$ and an extra 5 for the case $(q_i,1,1,L,q_j)$. To ensure that the instruction sequence is the same length for each case we append the length-10 sequence $(L,R,)^5$ to the sequence for case $(q_i,1,1,L,q_j)$. Satisfying restriction~3 in Equation~\eqref{eq:instructions sequence for read 0 transition rule} gives an extra 20 instructions for the cases $(q_i,0,0,R,q_j)$ and $(q_i,0,1,R,q_j)$, and an extra 10 for cases $(q_i,0,0,L,q_j)$ and $q_i,0,1,L,q_j$. To ensure that the instruction sequence is the same length for each case we append the length-10 sequence $(L,R,)^5$ to the sequences for case $(q_i,0,0,L,q_j)$ and case $(q_i,0,1,L,q_j)$. Now the length of the sequence that encodes each state is 54 (instead of 13), and so we replace jumps of the from $J(13i)$ with jumps of the from $J(54i)$. The sequence $R,J(13i+8),\tne{TR_{q_i,0}}$ of length 8 has been replaced by a sequence of length 32, and so we replace jumps of the from $J(13i+8)$ with jumps of the form $J(54i+32)$. This completes our conversion to a Wang~B machine with the 3 restrictions mentioned above. 
    \end{proof}    

  \section*{Acknowledgements}
  We thank Benedikt L\"{o}we for inviting us to the  Isaac Newton Institute for Mathematical Sciences in Cambridge (UK) to participate in the 2012 programme ``Semantics and Syntax: A
  Legacy of Alan Turing'', and for facilitating this fun project.  Finally, we thank David Soloveichik for
  translating the work of Zykin~\cite{Zykin1963}.

    \ifelsevierversion
\else
      \bibliographystyle{plain} 
    \fi
    \bibliography{has}
\end{document}